\documentclass[a4paper,reqno,11pt,oneside]{amsart}

\usepackage[utf8]{inputenc}
\usepackage[T1]{fontenc}
\usepackage{amsmath,amssymb,amsthm}
\usepackage{graphicx}
\usepackage{booktabs}
\usepackage{algorithm}
\usepackage{algpseudocode}
\usepackage{pgfplots}
\pgfplotsset{compat=1.18}
\usepackage[margin=1in]{geometry}
\usepackage{hyperref}
\usepackage[numbers]{natbib}
\usepackage{subfigure}
\usepackage{comment}
\usepackage[foot]{amsaddr}
\usepackage{placeins}

\newcommand{\Z}{\mathbb{Z}}

\newcommand{\E}{\mathbb{E}}

\newcommand{\lc}{\lambda_c}

\newcommand{\ind}{\mathbf{1}}
\newcommand{\CRS}{\textsc{CRS}}
\newcommand{\MT}{\textsc{MT}}
\newcommand{\MTo}{\textsc{MT}\ensuremath{_{0}}}

\newtheorem{conjecture}{Conjecture}
\newtheorem{theorem}{Theorem}

\algrenewcommand\algorithmicrequire{\textbf{Input:}}
\algrenewcommand\algorithmicensure{\textbf{Output:}}

\usepackage{tikz-network}
\usetikzlibrary{arrows.meta,decorations.pathmorphing,backgrounds,automata,trees,shapes}

\title[Stochastic rumors on the one-dimensional integer lattice]{Critical behavior of stochastic rumors\\ on the one-dimensional integer lattice}

\author{Nevena Mari\'c$^{1}$}
\thanks{$^*$Corresponding author.} 
\address{$^1$ Nevena Mari\'c and Aleksandar Mini\'c: School of Computing, Union University, Belgrade, Serbia. E-mail address: \url{nmaric@raf.rs}, \url{ aminic5317rn@raf.rs}}

\author{Aleksandar Mini\'c$^1$}

\author{Pablo M.\ Rodriguez$^{2,*}$}
\address{$^2$ Pablo M. Rodriguez: Center for Exact and Natural Sciences, Federal University of Pernambuco, Recife, PE, Brazil. E-mail address: \url{pablo@de.ufpe.br}}

\date{\today}

\begin{document}

\keywords{Multi-agent systems, Rumor spreading, Lattice, Criticality, Phase-transition} 

\begin{abstract}
We study rumor-spreading models on the one-dimensional integer lattice, with particular emphasis on the effect of forgetting mechanisms on the long-term behavior of the rumor. We consider the spatial version of the classical Maki--Thompson model together with two variants that incorporate forgetting. Each model is formulated as an interacting particle system and describes the propagation of information through local interactions between neighboring individuals. Our results show that the introduction of forgetting fundamentally changes the behavior of the process. In the classical Maki--Thompson model, the rumor becomes extinct, whereas the two models with forgetting exhibit nontrivial phase transitions between extinction and survival. More precisely, we establish the existence of critical transmission rates that separate regimes in which the rumor dies out almost surely from regimes in which it survives with positive probability. We also investigate the dependence of these critical thresholds on the stifling rate. Combining rigorous arguments with numerical simulations, we prove the existence of the phase transitions, obtain estimates for the critical values, and explore the corresponding phase diagrams. These findings show that forgetting, despite acting as a mechanism that removes spreaders, can qualitatively alter the dynamics of rumor propagation by creating regimes of sustained information spread. Our work provides new insight into the critical behavior of spatial rumor models and raises several questions for future research.
\end{abstract}

\maketitle

\section{Introduction}\label{sec:intro}

\subsection{Rumor spreading on graphs}

The mathematical and computational modeling of rumor spreading has attracted considerable attention, as evidenced by the variety of approaches and open problems in the area, \cite{FerrazDeArruda2023,ghosh2023,mou2025,xia2024,yang2017}. The literature is extensive and encompasses several research directions, including the estimation of rumor duration and the study of phase-transition phenomena arising from the dynamical behavior of the underlying processes, \cite{CRS2012,EF2023,Puerres2025,Svensson1993}.

Among the first mathematical models of spreading rumors are those introduced in \cite{DaleyKendall1965,MakiThompson1973}, which provide a framework for describing the dissemination of information in a population. These models are epidemic-like in nature and can be viewed as Markov processes on graphs. The model introduced by \cite{MakiThompson1973}, commonly referred to as the Maki--Thompson rumor model, represents the population by a complete graph whose vertices correspond to individuals and whose edges indicate potential interactions through which information may be transmitted. Each individual is assumed to be in one of three states: ignorant, spreader, or stifler. Ignorants are individuals who do not know the information, spreaders are individuals who know the information and actively disseminate it, and stiflers are individuals who know the information but no longer participate in its propagation.

In the Maki--Thompson model, the rumor spreads through contacts initiated by spreaders. When a spreader contacts an ignorant individual, the information is transmitted, and the ignorant becomes a spreader. When a spreader contacts an already informed individual, either a spreader or a stifler, the initiating spreader becomes a stifler due to a loss of interest in further spreading the rumor. More generally, when the model is defined on an arbitrary graph, each spreader may contact only its neighboring vertices, and the resulting dynamics are naturally described as an interacting particle system. Figure \ref{fig:MT} illustrates the possible transitions of the model.

\begin{figure}
    \centering
\begin{tikzpicture}

\node at (-2.5,0) {(a)};

\node at (0,0) {
\begin{tikzpicture}
\node at (1,-0.5) {\footnotesize $-1$};
\node at (2,-0.5) {\footnotesize $0$};
\node at (3,-0.5) {\footnotesize $1$};
\draw[dashed, ->] (2,0.3) to [out=90, in=135] (2.9,0.3); 
    \draw[dotted] (0.3,0) -- (0.5,0);
        \draw[dotted] (3.5,0) -- (3.7,0);
    \draw (0.6,0) -- (3.4,0);
    \filldraw [blue!30] (1,0) circle (7pt);
\node[white] at (1,0) {\footnotesize $0$};
    \filldraw [red!80] (2,0) circle (7pt);
    \node[white] at (2,0) {\footnotesize $1$};
    \filldraw [blue!30] (3,0) circle (7pt);
    \node[white] at (3,0) {\footnotesize $0$};
\end{tikzpicture}};

\draw[-{Stealth[length=2mm]}] (2.7,0.05) to (3.3, 0.05);

\node at (6,0) {
\begin{tikzpicture}
\node at (1,-0.5) {\footnotesize $-1$};
\node at (2,-0.5) {\footnotesize $0$};
\node at (3,-0.5) {\footnotesize $1$};
\draw[dashed, white, ->] (2,0.3) to [out=90, in=135] (2.9,0.3); 
    \draw[dotted] (0.3,0) -- (0.5,0);
        \draw[dotted] (3.5,0) -- (3.7,0);
    \draw (0.6,0) -- (3.4,0);
    \filldraw [blue!30] (1,0) circle (7pt);
\node[white] at (1,0) {\footnotesize $0$};
    \filldraw [red!80] (2,0) circle (7pt);
    \node[white] at (2,0) {\footnotesize $1$};
    \filldraw [red!80] (3,0) circle (7pt);
    \node[white] at (3,0) {\footnotesize $1$};
\end{tikzpicture}};

\node at (-2.5,-2) {(b)};
\node at (0,-2) {
\begin{tikzpicture}
\node at (1,-0.5) {\footnotesize $-1$};
\node at (2,-0.5) {\footnotesize $0$};
\node at (3,-0.5) {\footnotesize $1$};
\draw[dashed,  <-] (2.1,0.3) to [out=45, in=90] (3,0.3); 
    \draw[dotted] (0.3,0) -- (0.5,0);
        \draw[dotted] (3.5,0) -- (3.7,0);
    \draw (0.6,0) -- (3.4,0);
    \filldraw [blue!30] (1,0) circle (7pt);
\node[white] at (1,0) {\footnotesize $0$};
    \filldraw [red!80] (2,0) circle (7pt);
    \node[white] at (2,0) {\footnotesize $1$};
    \filldraw [red!80] (3,0) circle (7pt);
    \node[white] at (3,0) {\footnotesize $1$};
\end{tikzpicture}};

\draw[-{Stealth[length=2mm]}] (2.7,-1.95) to (3.3, -1.95);

\node at (6,-2) {
\begin{tikzpicture}
\node at (1,-0.5) {\footnotesize $-1$};
\node at (2,-0.5) {\footnotesize $0$};
\node at (3,-0.5) {\footnotesize $1$};
\draw[dashed, white, ->] (2,0.3) to [out=90, in=135] (2.9,0.3); 
    \draw[dotted] (0.3,0) -- (0.5,0);
        \draw[dotted] (3.5,0) -- (3.7,0);
    \draw (0.6,0) -- (3.4,0);
    \filldraw [blue!30] (1,0) circle (7pt);
\node[white] at (1,0) {\footnotesize $0$};
    \filldraw [red!80] (2,0) circle (7pt);
    \node[white] at (2,0) {\footnotesize $1$};
    \filldraw [orange] (3,0) circle (7pt);
    \node[white] at (3,0) {\footnotesize $2$};
\end{tikzpicture}};

\node at (-2.5,-4) {(c)};
\node at (0,-4) {
\begin{tikzpicture}
\node at (1,-0.5) {\footnotesize $-1$};
\node at (2,-0.5) {\footnotesize $0$};
\node at (3,-0.5) {\footnotesize $1$};
\draw[dashed, ->] (2,0.3) to [out=90, in=135] (2.9,0.3); 
    \draw[dotted] (0.3,0) -- (0.5,0);
        \draw[dotted] (3.5,0) -- (3.7,0);
    \draw (0.6,0) -- (3.4,0);
    \filldraw [blue!30] (1,0) circle (7pt);
\node[white] at (1,0) {\footnotesize $0$};
    \filldraw [red!80] (2,0) circle (7pt);
    \node[white] at (2,0) {\footnotesize $1$};
    \filldraw [orange] (3,0) circle (7pt);
    \node[white] at (3,0) {\footnotesize $2$};
\end{tikzpicture}};

\draw[-{Stealth[length=2mm]}] (2.7,-3.95) to (3.3, -3.95);

\node at (6,-4) {
\begin{tikzpicture}
\node at (1,-0.5) {\footnotesize $-1$};
\node at (2,-0.5) {\footnotesize $0$};
\node at (3,-0.5) {\footnotesize $1$};
\draw[dashed, white, ->] (2,0.3) to [out=90, in=135] (2.9,0.3); 
    \draw[dotted] (0.3,0) -- (0.5,0);
        \draw[dotted] (3.5,0) -- (3.7,0);
    \draw (0.6,0) -- (3.4,0);
    \filldraw [blue!30] (1,0) circle (7pt);
\node[white] at (1,0) {\footnotesize $0$};
    \filldraw [orange] (2,0) circle (7pt);
    \node[white] at (2,0) {\footnotesize $2$};
    \filldraw [orange] (3,0) circle (7pt);
    \node[white] at (3,0) {\footnotesize $2$};
\end{tikzpicture}};

\draw[dashed] (-2,1) to (8,1);

\node at (3,2) {

\begin{tikzpicture}
\filldraw [blue!30] (0,0) circle (7pt);
\node[white] at (0,0) {\footnotesize $0$};
\node at (1,0) {\footnotesize -\,ignorant};

\filldraw [red!80] (3,0) circle (7pt);
\node[white] at (3,0) {\footnotesize $1$};
\node at (4,0) {\footnotesize -\,spreader};

\filldraw [orange] (6,0) circle (7pt);
\node[white] at (6,0) {\footnotesize $2$};
\node at (6.8,0) {\footnotesize -\,stifler};
\end{tikzpicture}
};

\end{tikzpicture}

    \caption{Possible transitions of the Maki--Thompson model on $\mathbb{Z}$. Each vertex can be in one of three states: ignorant, spreader, or stifler, represented by $0$, $1$, and $2$, respectively, in the stochastic process. At time $0$, it is assumed that there is a single spreader located at vertex $0$, while all other vertices are ignorants. (a) The initial spreader contacts one of its two neighbors and successfully transmits the rumor. (b)–(c) The rumor continues to evolve. Whenever a spreader contacts a non-ignorant neighbor, it becomes a stifler.}
    \label{fig:MT}
\end{figure}
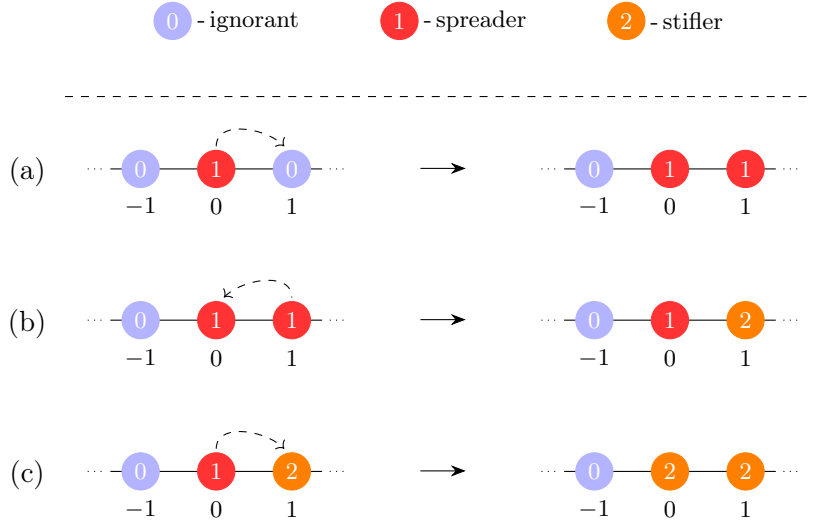

One of the main quantities studied in rumor models is the duration of the rumor; see, for example, \cite{Svensson1993}. Another important quantity is the asymptotic size of the rumor, which measures the proportion of individuals reached by the information by the end of the process; see \cite{LMR2011,Rada2021} for an overview of recent results. These studies are conducted under the assumption that the population is homogeneously mixed.

A second line of research considers the spread of rumors in heterogeneous populations, such as random graphs and complex networks. In this setting, the main objectives are the identification of relevant quantities and the study of phase-transition phenomena associated with the survival or extinction of the rumor. Most results in this direction rely on computational simulations and mean-field approximation techniques. See, for example, \cite{Moreno2004,Nekovee2007} and the references therein for further discussion.

An intermediate setting between these two approaches is provided by regular graphs, particularly lattices. Studying rumor propagation on such structures allows one to obtain rigorous results while retaining some of the spatial features present in heterogeneous networks. Significant progress has been made in the study of rumor models on (random) trees; see \cite{junior-rodriguez-speroto2020,junior-rodriguez-speroto2021}. More recently, these results have been extended to inhomogeneous trees; see \cite{Puerres2025}. Trees play a central role in this context, because they arise locally in many random graph models.

The spread of rumors on other regular lattices has received comparatively less attention. In particular, the only work known to us that studies rumor models of the type introduced in \cite{MakiThompson1973} on the $d$-dimensional hypercubic lattice is \cite{CRS2012}. The analysis of rumor processes on these structures is considerably more challenging due to the presence of cycles, which induce dependencies between different propagation paths. Nevertheless, \cite{CRS2012} establishes conditions under which the rumor becomes extinct or survives with positive probability. While this work provides a broad description of the model's behavior in arbitrary dimension, it does not address the critical thresholds or the phase diagram of the process.

The present work contributes to this line of research by studying three rumor models on the one-dimensional integer lattice. Our results provide a more detailed understanding of their critical behavior and raise several questions for future investigation.

\subsection{Stochastic rumors on the one-dimensional integer lattice}

A central interest in the study of interacting particle systems is the identification of phase transitions and the localization of critical parameters that separate qualitatively different dynamical regimes. Determining such thresholds is often a difficult problem on general graphs, even for well-studied models. One-dimensional structures therefore occupy a distinguished place in theory, as they frequently provide the simplest setting in which rigorous results can be established. Classical examples include the contact process \cite{durrett1980,durrett1988,Liggett1999}, the voter model \cite{clifford1973,holley1975,Liggett1999}, and related growth and spreading models. Studies on one-dimensional lattices have led to a detailed understanding of survival and extinction phenomena, invariant measures, convergence properties, and critical behavior, thereby revealing fundamental mechanisms that are often difficult to isolate in more complex geometries.

Beyond their intrinsic mathematical interest, one-dimensional results frequently serve as a guide for the investigation of the same processes on more general graphs. Exact computations and sharp estimates for critical parameters provide valuable benchmarks against which results on higher-dimensional lattices, trees, random graphs, and complex networks may be compared. Such findings often suggest conjectures, indicate the appropriate scaling behavior, and help identify the structural features governing the emergence of distinct phases. Even when quantitative predictions do not directly extend to other graph classes, the qualitative insights obtained from one-dimensional systems often remain relevant. Consequently, the rigorous analysis of phase transitions in one-dimensional interacting particle systems constitutes an important step toward understanding the behavior of stochastic spatial models on general networks.

Rumor spreading processes provide a notable example of this phenomenon. A large body of work has been devoted to rumor models on one-dimensional structures, particularly on the integer lattice, yielding a rich collection of results concerning survival, extinction, propagation speed, and asymptotic behavior; see, for instance, \cite{EF2023,EF2020,EFM2015,roy2024} and the references therein. Although the present work focuses on a different setting, the techniques developed here suggest a broader applicability. In particular, with additional technical effort, they could potentially be adapted to investigate analogous questions for rumor processes on higher-dimensional lattices, including the $d$-dimensional hypercubic lattice, as well as on other important graph classes such as regular trees and $k$-regular ring lattices. Thus, our results contribute to the ongoing effort to understand how geometric and topological features of a network influence the critical behavior and long-term dynamics of information spreading processes.

\subsection{Main contributions and organization of the paper}

We study three spatial stochastic rumor models on the one-dimensional integer lattice $\mathbb{Z}$: the classical Maki--Thompson model (\MTo), with permanent stiflers and no forgetting, and two variants that incorporate a forgetting mechanism for both spreaders and stiflers. The first is the Coletti--Rodriguez--Schinazi model (\CRS), proposed by \cite{CRS2012}, in which spreaders become stiflers only through contact with other spreaders. The second is a more natural variant of the classical model, which we call the spatial Maki--Thompson model (\MT), where both spreaders and stiflers induce stifling. For the three models we assume that transmissions occur at rate $\lambda$. In the \CRS{} and \MT{} models, stifling events occur at rate $\alpha$, while in the \MTo{} model the parameter $\alpha$ is absorbed into the time scale.

We prove that both \CRS{} and \MT{} models have a positive, finite critical transmission rate, and give numerical evidence that the \CRS{} critical line lies below the \MT{} one.  In addition, using an exact continuous-time simulation (Gillespie algorithm) with dynamically expanding lattices, we estimate the critical transmission rate $\lc(\alpha)$ as a function of the stifling rate $\alpha$ for each of these two models.

For the \CRS{} model we cover $\alpha\in[0,20]$, spanning estimated critical values from $\lc \approx 1.65$ (the contact-process limit) to $\lc\approx 2750$, and find that, for large $\alpha$, the critical curve is well described by the quadratic law $\lc(\alpha)\approx 1.65 + 7\alpha^2$. We give a heuristic argument for the quadratic scaling based on mutual parent--child stifling at the rumor front. For the \MT{} variant our simulations are consistent with a phase transition at every $\alpha\in[0,1)$, with critical rate well fit by $\lc(\alpha)\approx 1.65 + 2.8\,\alpha/(1-\alpha)^{2.65}$, growing steeply as $\alpha\uparrow 1$. 

For the classical \MTo{} model, we prove that the rumor dies out for every $\lambda$. We then illustrate this theoretical result through simulations, which consistently exhibit extinction across the entire range of $\lambda$ considered. Moreover, the mean total range satisfies $\E[N]=2(\lambda+1) \approx 2\lambda$, revealing a clean linear law consistent with the inescapable effect of the permanent-stifler wake in one dimension.

The paper is organized as follows. Section 2 contains the theoretical analysis of the models. After discussing their mean-field (deterministic) counterparts, we present the main results concerning extinction for the \MTo{} model and the existence of phase transitions for the \CRS{} and \MT{} models. Section 3 presents the simulation study and numerical results, which provide insight into the behavior of the critical parameters and motivate several conjectures. Section 4 concludes with a discussion of open questions and future research directions.

\section{Models description and theoretical results}\label{sec:models}
\subsection{Mathematical formulation}

\subsubsection{The Markov processes}
Each of the three models we consider is a continuous-time Markov process $(\eta_t)_{t\geq 0}$ with state space $\mathcal{S}:=\{0,1,2\}^{\Z}$. That is, at time $t$ the state of the process is a function $\eta_t: V \longrightarrow \{0,1,2\}$. Given a configuration $\eta \in \mathcal{S}$, at each site $x\in\Z$, the state $\eta(x)\in\{0,1,2\}$ represents an ignorant, a spreader, or a stifler, respectively. Writing
\[
n_1(x,\eta)=\sum_{y:\,|y-x|=1}\ind\{\eta(y)=1\},\qquad
n_2(x,\eta)=\sum_{y:\,|y-x|=1}\ind\{\eta(y)=2\}
\]
for the numbers of nearest neighbors of $x$ in states 1 and 2, the local
transition rates $c_{ij}(x,\eta)$ from state $i$ to state $j$ at site $x$, for each of the considered rumor models,
are summarized in Table~\ref{tab:rates}. In the rest of the paper, we denote by $(\eta_t^{\text{o}\lambda})_{t\geq 0}$, $(\eta_t^{*\lambda,\alpha})_{t\geq 0}$, and $(\eta_t^{\lambda,\alpha})_{t\geq 0}$ the Markov processes associated with the \MTo{}, \CRS{}, and \MT{} models, respectively. We assume that the three processes start with one spreader at $0$ and all remaining vertices being ignorants. That is, we consider the initial configuration $\eta_0 \in \mathcal{S}$ such that $\eta_0(0)=1$ and $\eta_0(x)=0$ for all $x\neq 0$. We refer to this as the standard initial configuration. 

\begin{table}[h]
\centering
\caption{Transition rates for the three models. Parameters:
$\lambda>0$ (spreading rate), $\alpha\geq 0$ (stifling rate).}
\label{tab:rates}
\begin{tabular}{lccccc}
\toprule
Model & Notation & $c_{01}(x,\eta)$ & $c_{12}(x,\eta)$ & $c_{10}(x,\eta)$ & $c_{20}(x,\eta)$ \\
\midrule
\MTo & $(\eta_t^{\text{o},\lambda})_{t\geq 0}$ &$\lambda\,n_1$ & $n_1+n_2$             & $0$ & $0$ \\[.2cm]
\CRS & $(\eta_t^{*\lambda,\alpha})_{t\geq 0}$ &$\lambda\,n_1$ & $\alpha\,n_1$         & $1$ & $1$ \\[.2cm]
\MT  & $(\eta_t^{\lambda,\alpha})_{t\geq 0}$ &$\lambda\,n_1$ & $\alpha\,(n_1+n_2)$   & $1$ & $1$ \\
\bottomrule
\end{tabular}
\end{table}

The \MTo{} model is a spatial version on $\mathbb{Z}$ of the classical Maki--Thompson dynamics \citep{MakiThompson1973}: stifling occurs at unit rate from both spreaders and stiflers, and neither spreaders nor stiflers forget the rumor. Therefore, once a site enters state $2$, it remains there forever. The other two models incorporate a forgetting mechanism. On one hand, the \CRS{} model is the spatial rumor model introduced in~\cite{CRS2012}: a spreader is stifled only through contact with other spreaders, which occurs at rate $\alpha$, and both spreaders and stiflers forget the rumor at rate $1$. This model was proposed in $\mathbb{Z}^d$, for $d \geq 1$, as a spatial variant of the Maki--Thompson model that incorporates a forgetting mechanism. In \cite{CRS2012}, the authors obtained sufficient conditions on the rates under which the rumor either becomes extinct or survives with positive probability. On the other hand, the \MT{} model is a more natural spatial variant of the Maki--Thompson model because, as in the classical model, stiflers as well as spreaders induce stifling. This occurs at rate $\alpha$. Here, we also assume that both spreaders and stiflers forget the rumor at rate $1$. Note that in the \MTo{} model, the rate parameter $\alpha$ has been absorbed into the time scale.

\subsubsection{Harris’ graphical construction}

Thus defined, the rumor models are interacting particle systems, so we can define them through Harris’ graphical construction. To this end, consider a collection of independent Poisson processes 
$$\{N^{x,y}_{\lambda}, N^{x,y}_{\alpha}, D^x : x,y\in \mathbb{Z}, ||x-y||=1\},$$ 
with respective intensities $\lambda$, $\alpha$, and $1$. Then, at each arrival time of $N^{x,y}_{\lambda}$, if sites $x$ and $y$ are in states $1$ and $0$, respectively, in any of the three models, the state of site $y$ is updated to $1$. On the other hand, the stifling rule is governed by the processes $N^{x,y}_{\alpha}$ and depends on the model. At each arrival time of $N^{x,y}_{\alpha}$, in the \CRS{} model the state of $x$ changes to $2$ if both $x$ and $y$ are in state $1$, whereas in the \MT{} model it changes to $2$ if $x$ is in state $1$ and $y$ is in state $1$ or $2$. In the \MTo{} model the stifling rule is the same as in the \MT{} model but acts at unit rate rather than at rate $\alpha$, since for \MTo{} the parameter $\alpha$ is absorbed into the time scale; this corresponds to taking the stifling intensity equal to $1$ (that is, $N^{x,y}_{\alpha}$ with $\alpha=1$). Forgetting occurs only in the \CRS{} and \MT{} models: at each arrival time of $D^x$, if vertex $x$ is in state $1$ or $2$, its state changes to $0$. In the \MTo{} model there is no forgetting, so the marks of $D^x$ are ignored and stiflers are permanent.
In this way, we obtain realizations of the three stochastic rumor models on the one-dimensional integer lattice with transition rates given in Table \ref{tab:rates}. Note that, in order to construct the process inside a finite space--time box, it suffices to consider only the Poisson arrival times occurring within that box. For further details on graphical constructions, we refer the reader to \cite{Durrett1995,Harris1972}. The graphical construction is particularly useful because it allows the three Markov processes to be constructed on the same probability space, which will be important later for comparison purposes. See Figure \ref{fig:HGR}.

\begin{figure}[ht!]
\subfigure[Realization of the \CRS{} rumor model.]{
    \begin{tikzpicture}[
scale=1,
>=Stealth,
ignorant/.style={line width=0.5pt},
spreader/.style={line width=1.5pt},
stifler/.style={line width=2pt},
arrow/.style={->,line width=0.6pt}
]

\draw[->] (-1,0) -- (-1,8);
\node[rotate=90] at (-1.5,4) {Time};

\draw (-0.5,0) -- (5.5,0);



\draw[ignorant] (0,0) -- (0,4.2);
\draw[spreader] (0,4.2) -- (0,6);
\draw[ignorant] (0,6) -- (0,8);


\draw[ignorant] (1,0) -- (1,1);
\draw[spreader] (1,1) -- (1,5.2);
\draw[dotted,stifler] (1,5.2) -- (1,8);


\draw[spreader] (2,0) -- (2,7.5);
\draw[ignorant] (2,7.5) -- (2,8);


\draw[spreader] (3,0) -- (3,4.5);
\draw[dotted,stifler] (3,4.5) -- (3,8);


\draw[spreader] (4,2.5) -- (4,8);
\draw[ignorant] (4,0) -- (4,2.5);

\draw[ignorant] (5,0) -- (5,7.1);
\draw[spreader] (5,7.1) -- (5,8);


\draw[->,decorate] (1,4.2) -- (0,4.2);
\draw[->,decorate,dashed] (1,5.2) -- (2,5.2);
\draw[->,decorate,dashed] (2,6.5) -- (3,6.5);
\draw[arrow] (2,1) -- (1,1);
\draw[arrow] (3,2.5) -- (4,2.5);
\draw[->,decorate,dashed] (3,4.5) -- (4,4.5);
\draw[->,decorate,dashed] (4,6) -- (3,6);
\draw[->,decorate] (4,7.1) -- (5,7.1);
\draw[arrow] (5,2.2) -- (4,2.2);


 \draw [black] (0,6) circle (2.5pt);
    \filldraw [white] (0,6) circle (2.1pt);
 \draw [black] (2,7.5) circle (2.5pt);
    \filldraw [white] (2,7.5) circle (2.1pt);
 \draw [black] (4,1.3) circle (2.5pt);
    \filldraw [white] (4,1.3) circle (2.1pt);


\node at (0,-0.5) {0};
\node at (1,-0.5) {0};
\node at (2,-0.5) {1};
\node at (3,-0.5) {1};
\node at (4,-0.5) {0};
\node at (5,-0.5) {0};

\end{tikzpicture}}\qquad
\subfigure[Realization of the \MT{} rumor model.]{
    \begin{tikzpicture}[
scale=1,
>=Stealth,
ignorant/.style={line width=0.5pt},
spreader/.style={line width=1.5pt},
stifler/.style={line width=2pt},
arrow/.style={->,line width=0.6pt}
]

\draw[->] (-1,0) -- (-1,8);
\node[rotate=90] at (-1.5,4) {Time};

\draw (-0.5,0) -- (5.5,0);



\draw[ignorant] (0,0) -- (0,4.2);
\draw[spreader] (0,4.2) -- (0,6);
\draw[ignorant] (0,6) -- (0,8);


\draw[ignorant] (1,0) -- (1,1);
\draw[spreader] (1,1) -- (1,5.2);
\draw[dotted,stifler] (1,5.2) -- (1,8);


\draw[spreader] (2,0) -- (2,6.5);
\draw[dotted,stifler] (2,6.5) -- (2,7.5);
\draw[ignorant] (2,7.5) -- (2,8);

\draw[spreader] (3,0) -- (3,4.5);
\draw[dotted,stifler] (3,4.5) -- (3,8);

\draw[ignorant] (4,0) -- (4,2.5);
\draw[spreader] (4,2.5) -- (4,6);
\draw[dotted,stifler] (4,6) -- (4,8);

\draw[ignorant] (5,0) -- (5,8);


\draw[->,decorate] (1,4.2) -- (0,4.2);
\draw[->,decorate,dashed] (1,5.2) -- (2,5.2);
\draw[->,decorate,dashed] (2,6.5) -- (3,6.5);
\draw[arrow] (2,1) -- (1,1);
\draw[arrow] (3,2.5) -- (4,2.5);
\draw[->,decorate,dashed] (3,4.5) -- (4,4.5);
\draw[->,decorate,dashed] (4,6) -- (3,6);
\draw[->,decorate] (4,7.1) -- (5,7.1);
\draw[arrow] (5,2.2) -- (4,2.2);


 \draw [black] (0,6) circle (2.5pt);
    \filldraw [white] (0,6) circle (2.1pt);
 \draw [black] (2,7.5) circle (2.5pt);
    \filldraw [white] (2,7.5) circle (2.1pt);
 \draw [black] (4,1.3) circle (2.5pt);
    \filldraw [white] (4,1.3) circle (2.1pt);


\node at (0,-0.5) {0};
\node at (1,-0.5) {0};
\node at (2,-0.5) {1};
\node at (3,-0.5) {1};
\node at (4,-0.5) {0};
\node at (5,-0.5) {0};

\end{tikzpicture}}

    \caption{The graphical representation of the CRS rumor model (a) and the MT rumor model (b) on $\mathbb{Z}$ based on the same collection of independent Poisson processes. The horizontal line represents space, and the vertical line above a point $x$ is the time-line at $x$. The lines, bold lines and dotted bold lines indicate the portions of space–time that vertices are ignorant, spreader, and stifler, respectively. Arrows indicate transmission events associated with the processes $N^{x,y}_{\lambda}$, at which an ignorant becomes a spreader. Dashed arrows represent stifling interactions associated with the processes $N^{x,y}_{\alpha}$, causing a spreader to become a stifler. Finally, marks $\circ$ denote forgetting events associated with the processes $D^x$, corresponding to transitions from spreader or stifler states to the ignorant state.}
    \label{fig:HGR}
\end{figure}
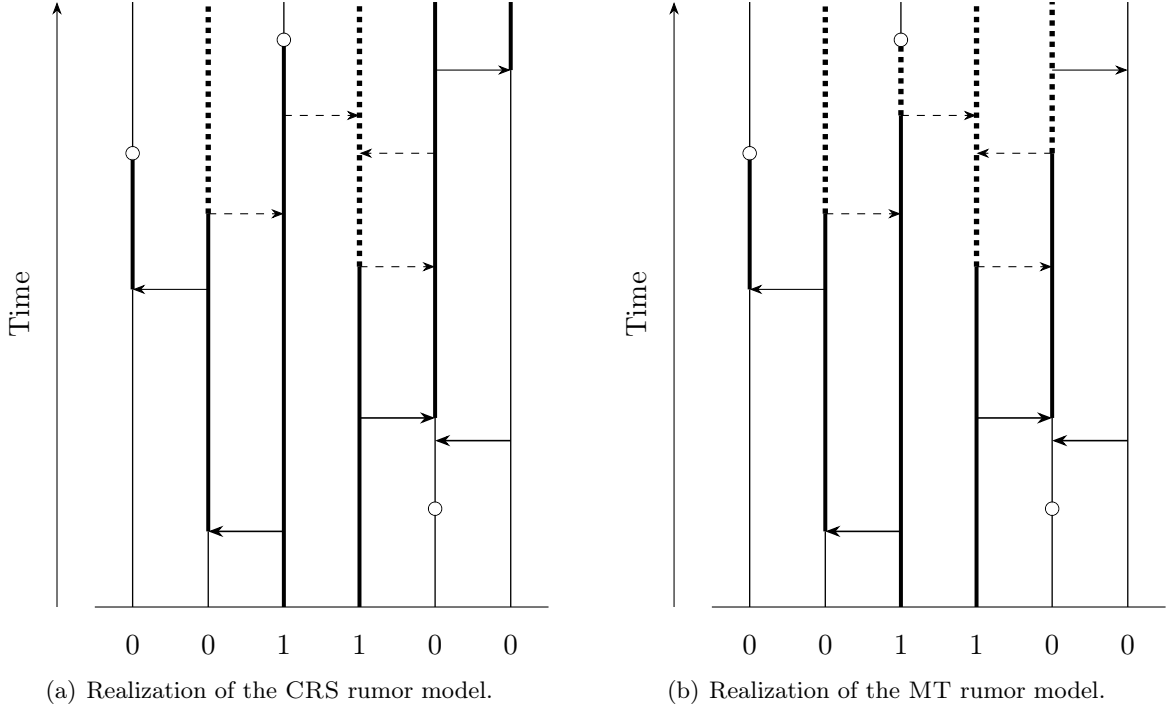

\subsection{The homogeneously mixed population setting}

The classical Maki--Thompson model was introduced by \cite{MakiThompson1973} in the homogeneously mixed population setting, which is equivalent to assuming that the population is represented by a complete graph. In this case, the deterministic version of the model, i.e., its mean-field approximation, is given by the dynamical system

\begin{equation}\label{eq:MT0-ODE}
\left\{
\begin{array}{rcl}
x'(t) &=& -x(t)y(t),\\[.2cm] 
y'(t) &=& (2x(t)-1)y(t).
\end{array}
\right.
\end{equation}

\noindent
With the usual initial conditions $x(0)=1$ and $y(0)=0$, it is well-known that the solution of \eqref{eq:MT0-ODE} approximates, for sufficiently large $n$, a scaled version of the entire trajectories of the Maki–Thompson model \cite{Rada2021}. In this context, $x(t)$, $y(t)$ and $z(t)=1-x(t)-y(t)$ represent good approximations for the proportions of ignorants, spreaders and stiflers, respectively, at time $t$, in the stochastic process, provided $t>0$ and $n$ is large enough. This justifies why many results can be obtained directly from a suitable analysis of this limit dynamical system. The random fluctuations between the solution of \eqref{eq:MT0-ODE} and the stochastic trajectories of the original process have been extensively studied in the literature for the \MTo{} model and various generalizations, typically using martingale approaches or convergence techniques for density-dependent Markov chains, see \cite{LMR2011} and the references therein. These results are usually presented in the form of limit theorems.

The mean-field approximation to the \CRS{} dynamics yields the dynamical system 
\begin{equation}\label{eq:CRS-ODE}
\left\{
\begin{array}{rcl}
y'(t)&=&\lambda y(t)(1-y(t)-z(t))-\alpha y(t)^2-y(t),\\[.2cm]
z'(t)&=&\alpha y(t)^2-z(t),
\end{array}
\right.
\end{equation}
whose rumor-free equilibrium $(0,0)$ is unstable iff $\lambda>1$,
\emph{independently of $\alpha$}, see \citep{CRS2012}. On the other hand, it is not difficult to see that the mean-field approximation to the \MT{} dynamics yields the dynamical system

\begin{equation}\label{eq:MT-ODE}
\left\{
\begin{array}{rcl}
    y'(t)&=&\lambda y(t)(1-y(t)-z(t))-\alpha y(t)^2-\alpha y(t)z(t)-y(t),\\[.2cm]
z'(t)&=&\alpha y(t)^2+\alpha y(t)z(t)-z(t),
\end{array}
\right.
\end{equation}
for which $(0,0)$ is a steady state that we call the free rumor equilibrium state. Note that the Jacobian matrix of the system \eqref{eq:MT-ODE} at the steady state $(0,0)$ is

$$\left(
\begin{array}{cc}
    \lambda -1 & 0 \\
    0 & -1 
\end{array}
\right)$$
which is the same as the one of system \eqref{eq:CRS-ODE}. Therefore, the two eigenvalues are $\lambda -1$ and $-1$, and like in the \CRS{} model, the steady state $(0, 0)$ is unstable, if and only if, $\lambda > 1$; again, independently of $\alpha$.

As we will see in Section~\ref{sec:results}, these predictions for the \CRS{} and the \MT{} models fail qualitatively on $\Z$ since spatial structure makes a critical value of $\lambda$ strongly dependent on $\alpha$.

\subsection{Partial results on $\mathbb{Z}$ and first open questions} \label{sec:theorems}

For the three models, we say that the rumor \emph{survives} if the set of spreaders is nonempty for all times $t\geq 0$, and \emph{dies out} otherwise. We denote by $\theta^o(\lambda)$, $\theta^*(\lambda,\alpha)$, and $\theta(\lambda,\alpha)$ the survival probabilities of the \MTo{}, \CRS{}, and \MT{} models, respectively. In the context of a stochastic process, critical behavior refers to a qualitative change in the long-term behavior of the system when a parameter crosses a special threshold, called the critical point or critical value. In the case of rumor models on $\mathbb{Z}$, we study the spreading rate $\lambda$ and observe that the forgetting mechanism causes the process to behave very differently, in terms of survival or extinction, depending on whether $\lambda$ is below or above a certain critical value. Unlike in homogeneously mixed populations, this critical value depends strongly on the parameter $\alpha$.

\subsubsection{Extinction of the \MTo{} model and the size of the rumor} The proof that the \MTo{} model dies out almost surely, for any $\lambda >0$, is a consequence of the following result. 

\begin{theorem}\label{prop:MToZ}
 Let $\lambda >0$ and let $(\eta_{t}^{\text{o},\lambda})_{t\geq 0}$ be the \MTo{} model on $\mathbb{Z}$ with the standard initial configuration. Let $\tau_{\infty}=\inf\{t\geq 0: \eta_{t}^{\text{o},\lambda}(x)\neq 1, \text{ for all }x\in\mathbb{Z}\}$ be the absorption time of the process, and let 
    $$N(t):=\sum_{x\in \mathbb{Z}} 1_{\{\eta_t^{\text{o},\lambda}(x)=2\}},$$
    be the number of stiflers at time $t$. Then $\mathbb{P}(N(\tau_{\infty})<\infty)=1$, for all $\lambda >0$. Moreover, $\mathbb{E}(N(\tau_{\infty}))=2(\lambda +1)$.
\end{theorem}

\begin{proof}
    The result is a consequence of considering the random variables $R$ as the rightmost informed vertex when the process starts with vertex $1$ as a spreader, and $0$ as spreader or stifler, $L$ as the leftmost informed vertex starting from vertex $-1$ as a spreader, and $0$ as spreader or stifler, and $Y$ be equal to $-1,1,2$ according to only $-1$ is informed by $0$, only $1$ is informed by $0$ or both $-1,1$ are informed by $0$, respectively. Thus defined it is not difficult to see that $R\sim Geom(1/(\lambda +1))$, $L\sim  Geom(1/(\lambda +1))$ and $Y$ is such that $\mathbb{P}(Y=-1)=\mathbb{P}(Y=1)=\frac{1}{2(\lambda+1)}$ and $\mathbb{P}(Y=2)=1-\frac{1}{\lambda+1}$, and it is independent of the random variables $R$ and $L$. Therefore $N(\tau_{\infty})=1 + R\, \mathbb{I}_{\{Y\neq -1\}} + L\, \mathbb{I}_{\{Y\neq 1\}},$ and the result holds. Indeed, note that
\begin{equation}
\begin{array}{rcl}
\mathbb{E}(N(\tau_{\infty})) &= &1 + \mathbb{E}(R)\mathbb{P}(Y\neq -1) + \mathbb{E}(L) \mathbb{P}(Y\neq 1)\\[.2cm]
&=&1+2(\lambda +1)\left(1-\frac{1}{2(\lambda+1)}\right)\\[.2cm]
&=&2(\lambda +1).
\end{array}
\end{equation}
\end{proof}

Thus defined, $N(\tau_{\infty})$ is the total number of distinct vertices of $\mathbb{Z}$ ever occupied by a spreader during the process. Setting $\lambda = 1$, we recover the result proved in \cite[Proposition 2.1]{DiazBacca2026}. We also refer the reader to \cite{DiazBacca2026-preprint} for a review of similar results obtained for an asymmetric $k$-fold stifling Maki--Thompson model on $\mathbb{Z}$ and for the Maki--Thompson model under the majority rule on $\text{Cay}(\mathbb{Z},\{\pm 1,\pm 2\})$, the Cayley graph of degree $4$ on $\mathbb{Z}$. In those works, the authors applied their results to study the behavior of the same models on suitably defined ring lattices. Their proofs rely on coupling the processes on the ring lattices with their infinite-volume counterparts, an argument that is naturally formulated through a graphical construction of the processes.

\subsubsection{Critical transmission rates for the \CRS{} and the \MT{} models}

At $\alpha=0$ the \CRS{} and \MT{} models both reduce to the basic
nearest-neighbor contact process on $\Z$. For the contact process $(\xi_{t})_{t\geq 0}$ it is well known that the survival probability
$$\theta^{cp}(\lambda):=\mathbb{P}(\cap_{t\geq 0}\cup_{x\in \mathbb{Z}}\{\xi_{t}(x)\neq 0\}),$$
is non-decreasing as a function of $\lambda$, and if $\lambda_c^{cp}:=\inf\{\lambda>0:\theta^{cp}(\lambda)>0\}$, then $\lambda_c^{cp}\approx 1.6489$~\citep{Liggett1999}. Moreover, it trivially hold
$$\theta^*(\lambda,0)=\theta(\lambda,0)\left\{
\begin{array}{lc}
    =0 & \text{ if }\lambda < \lambda_c^{cp},\\[.2cm]
     >0 & \text{ if }\lambda > \lambda_c^{cp}.
\end{array}
\right.
$$

As far as we know, the only rigorous results obtained for the rumor models considered here, with $\alpha > 0$, are those established by \cite{CRS2012} for the \CRS{} model on $\mathbb{Z}^d$, $d \geq 1$. In that work, the authors obtained sufficient conditions under which the rumor either becomes extinct or survives with positive probability. More precisely, they proved the following theorem.

\begin{theorem}\label{thm:CRS}
\cite[Theorems 2.1 and 2.2]{CRS2012}    Consider the \CRS{} rumor model on $\mathbb{Z}$ with the standard initial configuration, and let $\lambda_c^{cp}$ the critical parameter for the contact process on $\mathbb{Z}$. 
    \begin{enumerate}
        \item[(i)] If $\lambda \leq \lambda_c^{cp}$, then the rumor becomes extinct for all $\alpha \geq 0$.
        \item[(ii)] Let $\lambda > \lambda_c^{cp}$ be fixed. Then there exists constants $0<\alpha_1(\lambda)<\alpha_2(\lambda)<\infty$ such that
        \begin{enumerate}
            \item If $\alpha < \alpha_1(\lambda)$, the rumor may survive;
            \item If $\alpha > \alpha_2(\lambda)$, the rumor becomes extinct.
        \end{enumerate}
    \end{enumerate}
\end{theorem}

Note that Theorem \ref{thm:CRS} provides only qualitative information about the behavior of the process. In particular, it does not establish the existence of a unique critical parameter or provide a more precise description of the behavior of the process when $\lambda$ depends on $\alpha$. For each fixed $\alpha > 0$, we are unable to prove that the process is monotone in $\lambda$. The usual approach used for the contact process, namely coupling through the graphical representation of the process, fails in this setting. Indeed, although increasing $\lambda$ increases the number of spreaders, a larger number of spreaders also leads to more stifling events, so the traditional coupling argument does not apply. As we will see in Section \ref{sec:survival}, simulations suggest that this monotonicity should hold, see Figure~\ref{fig:survival}. However, new techniques are required for a rigorous proof, and we leave this question for future work. 

For our study, we consider the critical transmission rates associated with a given stifling rate, defined by

\begin{equation}
\lc^*(\alpha) =\inf\{\lambda >0: \theta^*(\lambda,\alpha)>0\},\quad \text{and}\quad
\lc(\alpha) =\inf\{\lambda >0:\theta(\lambda,\alpha)>0\},
\end{equation}

\smallskip
\noindent
for the \CRS{} and \MT{} models, respectively.

\begin{theorem} \label{thm:phase}
Consider the \CRS{} and the \MT{} rumor models on $\mathbb{Z}$ with the standard initial configuration, and let $\lambda_c^*$ and $\lambda_c$ be its critical transmission rates, respectively. Then, for any $\alpha > 0$,
\begin{enumerate}
    \item[(i)] $0<\lambda_c^*(\alpha)<\infty.$
    \item[(ii)] $ 0< \lambda_c(\alpha)<\infty.$
\end{enumerate}
\end{theorem}

\begin{proof} To prove positivity fix $\alpha>0$ and note that by Theorem \ref{thm:CRS}(i), we have $\theta^*(\lambda,\alpha)=0$ for all $\lambda \leq \lambda_c^{cp}$. Therefore, $\lambda_c^{cp} \leq \lambda_c^*(\alpha)$, and the result follows since $\lambda_c^{cp}>0$. 
The same lower bound holds for the \MT{} model. Indeed, the set of \MT{} spreaders is dominated by a contact process: an ignorant site turns into a spreader only through a mark of $N^{x,y}_\lambda$ from a spreader neighbour (rate $\lambda$ per such neighbour), while every spreader leaves state $1$ at total rate $1+\alpha(n_1+n_2)\geq 1$. Coupling through the marks $N^{x,y}_\lambda$ and $D^x$ and ignoring those of $N^{x,y}_\alpha$ produces a contact process $\xi$ with $\xi_t(x)\geq \mathbf{1}\{\eta_t(x)=1\}$ for all $x$ and $t$; hence for $\lambda\leq\lambda_c^{cp}$ the \MT{} spreaders die out, so $\theta(\lambda,\alpha)=0$ and $\lambda_c(\alpha)\geq\lambda_c^{cp}>0$.

The proof of finiteness is constructive and is based on coupling arguments using the graphical construction of the process. Let $\Delta > 0$, and let $m,n \in \mathbb{Z}$ be such that $m+n$ is even. We shall define independent random variables $X_{m,n}$, taking the values $0$ and $1$. We declare $X_{m,n}=1$, and call
$(m,n)$ open, if and only if, in the graphical representation of the \MT{} rumor model, the following events occur:

\begin{itemize}
\item there is no forgetting events in the interval $\{m\}\times ((n-1)\Delta,(n+1)\Delta]$,

\item there is no stifling interactions from the interval $\{m\}\times ((n-1)\Delta,(n+1)\Delta],$

\item there exist left and right transmission events from the
interval $\{m\}\times (n\Delta,(n+1)\Delta].$
\end{itemize}

By construction we have that the random variables $X_{m,n}$ are independent, and
\[
p=p(\Delta)=\mathbb{P}_{\lambda}(X_{m,n}=1)
=e^{-2(1+\alpha)\Delta}(1-e^{-\lambda \Delta})^2.
\]

Now, we choose $\Delta$ to maximize $p(\Delta)$; that is $\Delta=\frac{1}{\lambda}\log\!\left(1+\frac{\lambda}{1+\alpha}\right).$ Thus,
\begin{equation}
\label{6.8}
p=p(\Delta)=\frac{(1+\alpha)^{2(1+\alpha)/\lambda}\lambda^2}{(\lambda+1+\alpha)^{2(\lambda + 1 + \alpha)/\lambda}}.
\end{equation}

Consider the $X_{m,n}$ as giving rise to a directed site percolation model on the first quadrant of a rotated copy of $\mathbb{Z}^2$. It is not difficult to see that, for any $n\in \mathbb{N}$, the set $V_n$ of vertices of the form $(m,n)$ that are reached from $(0,0)$ along open paths of the percolation process can be identified with a subset of vertices $m\in\mathbb{Z}$ in the rumor process such that $m$ is a spreader at time $(n+1)\Delta$. Note that $\mathbb{P}(\cap_{n\geq 0}\{|V_n|>1\})>0 \quad \text{if } p>\overrightarrow{p}_c^{site}$, where $\overrightarrow{p}_c^{\mathrm{site}}$ is the critical probability of the percolation model. Thus, we have that $\theta(\lambda,\alpha)>0$ if
$p(\Delta) > \overrightarrow{p}_c^{\mathrm{site}}$. Since $\overrightarrow{p}_c^{\mathrm{site}}<1$, see \cite{Ballister1994,Durrett1992,Liggett1995}, the final inequality is valid for sufficiently large $\lambda$, and we deduce that $ \lambda_c(\alpha)<\infty$.

The same construction yields finiteness for the \CRS{} model without modification. The event
$\{X_{m,n}=1\}$ is defined solely through the driving Poisson marks, requiring no marks of $D^x$ and no marks of $N^{x,y}_\alpha$ on $\{m\}\times((n-1)\Delta,(n+1)\Delta]$, together with left and right marks of $N^{x,y}_\lambda$ on $\{m\}\times(n\Delta,(n+1)\Delta]$. Forbidding the $N^{x,y}_\alpha$ marks prevents every stifling event of \emph{either} model on this region, since the \CRS{} stifling events form a subset of the \MT{} ones. An open path therefore carries a surviving, advancing spreader in the \CRS{} dynamics exactly as in the \MT{} dynamics, the probability $p(\Delta)$ in \eqref{6.8} is unchanged, and the same percolation comparison gives $\theta^*(\lambda,\alpha)>0$ for $\lambda$ large enough. Hence $\lambda_c^*(\alpha)<\infty$.
\end{proof}

At this point, a natural question is whether there is any connection between the two models. Intuitively, one may expect the \CRS{} model to dominate the \MT{} model, in the sense that
$\theta(\lambda,\alpha)\leq \theta^*(\lambda,\alpha)$, which would imply that $\lambda_c^{*}(\alpha)\leq \lambda_c(\alpha).
$ The intuition behind this conjecture is that the \CRS{} dynamics preserve state $1$ for a longer period of time, thereby favoring the persistence of the rumor. However, standard coupling arguments do not seem sufficient to establish this claim, suggesting that new approaches may be required. In Section~\ref{sec:results}, we present numerical evidence supporting this conjecture;  see Figure~\ref{fig:crs_vs_mt}, where the CRS line lies below the MT line throughout the measured range.


\section{Simulations and numerical results} \label{sec:results}
\subsection{Simulation method}\label{sec:method}

We simulate each of the three models by direct sampling of the
continuous-time Markov chain via the Gillespie
algorithm~\citep{Gillespie1977}. The two main technical ingredients are
an exact event-driven time stepping and a lattice that expands
dynamically whenever the rumor approaches its boundary.

\subsubsection{The Gillespie algorithm for spatial rumor models}

Let $r(x,\eta)$ denote the total rate of a transition at site $x$ in
configuration $\eta$, i.e.
\[
r(x,\eta) \;=\; \sum_{j\neq \eta(x)} c_{\eta(x)j}(x,\eta),
\]
and let $R(\eta)=\sum_x r(x,\eta)$ be the total rate of the process.
Algorithm~\ref{alg:gillespie} describes one realization of the process
until extinction, a prescribed event cap $M$, or a boundary-expansion
limit is reached.

\begin{algorithm}[h]
\caption{Gillespie simulation of a spatial rumor model on a finite lattice}
\label{alg:gillespie}
\begin{algorithmic}[1]
\Require model type $\in\{\CRS,\MT,\MTo\}$, parameters $\lambda,\alpha$;
         initial lattice size $L$; event cap $M$; expansion size $\Delta L$
\Ensure  trajectory observables: $N$ (total distinct sites visited by a
         spreader), $T$ (total time the origin spends in state 1),
         event count, final time
\State Initialize $\eta\gets 0$ on $\{-L/2,\dots,L/2\}$; set $\eta(0)\gets 1$
\State Compute $n_1(x,\eta), n_2(x,\eta)$ and $r(x,\eta)$ for all $x$
\State $t\gets 0$;\ $k\gets 0$;\ $N\gets 1$
\While{$\sum_x r(x,\eta) > 0$ \textbf{ and } $k<M$}
  \State $R \gets \sum_x r(x,\eta)$
  \State Draw $\tau\sim\mathrm{Exp}(R)$;\ \ $t\gets t+\tau$
  \State Draw site $x^\star$ with probability $r(x,\eta)/R$
  \State Determine transition at $x^\star$:
  \If{$\eta(x^\star)=0$}
      \State $\eta(x^\star)\gets 1$ \Comment{ignorant $\to$ spreader}
      \If{$x^\star$ not previously visited} $N\gets N+1$ \EndIf
  \ElsIf{$\eta(x^\star)=1$}
      \State $q \gets c_{12}(x^\star,\eta)/[c_{12}(x^\star,\eta)+c_{10}(x^\star,\eta)]$
      \State With prob.\ $q$ set $\eta(x^\star)\gets 2$, else $\eta(x^\star)\gets 0$
  \Else
      \State $\eta(x^\star)\gets 0$ \Comment{stifler $\to$ ignorant (\CRS, \MT) / impossible (\MTo)}
  \EndIf
  \State Update $n_1,n_2,r$ at $x^\star$ and its neighbors
  \If{a spreader exists within distance $2$ of a lattice endpoint}
     \State Expand lattice by $\Delta L$ ignorant sites on that side
  \EndIf
  \State $k\gets k+1$
\EndWhile
\State \Return $(N, T, k, t)$
\end{algorithmic}
\end{algorithm}

The algorithm is exact in the sense that it produces a realization of the
continuous-time Markov chain with the prescribed generator. Because only
sites with nonzero rate are active (ignorant sites without spreader
neighbors contribute $0$ to $R$), the event-driven sampling is efficient:
the cost of each event scales with the size of the active region around
the front rather than with the lattice size. We maintain the rate vector
incrementally, updating only the entries at the fired site and its
neighbors after each event.

\subsubsection{Dynamic lattice expansion}

A fixed large lattice would either waste memory or eventually bias the
dynamics through the boundary. We therefore start from a moderate initial
size (typically $L=2001$) and append $\Delta L=1000$ ignorant sites on
either end whenever a spreader reaches within two sites of the current
boundary. The number of such expansions is recorded as a diagnostic.

\subsubsection{Observables and the critical value}

For each parameter pair $(\lambda,\alpha)$ we perform $n$ independent runs
of Algorithm~\ref{alg:gillespie} and record:

\begin{itemize}
\item $N$: the total number of distinct sites ever occupied by a spreader
      during the run;
\item $T$: the total time the origin spends in state~1, integrated over the
      run;
\item the total number of events $k$ and the number of lattice expansions.
\end{itemize}

In the subcritical regime the rumor dies almost surely in finite time and
$N,T$ are finite random variables. As $\lambda\uparrow\lc(\alpha)$, both
$\E[N]$ and $\E[T]$ diverge. We estimate $\lc(\alpha)$ by a
\emph{polynomial-extrapolation} procedure, previously used in related
contexts \citep{GarciaMaric2008}: compute the sample mean $\widehat{\E[N]}(\lambda)$
on a grid of subcritical $\lambda$-values, form the inverse
$1/\widehat{\E[N]}(\lambda)$, fit a low-degree polynomial in $\lambda$, and
take $\lc$ to be the smallest positive root of that polynomial. The same
procedure applied to $\E[T]$ provides an independent cross-check.

\subsubsection{Validation against the contact process}

At $\alpha=0$ the \CRS{} and \MT{} models coincide with the basic contact process. Applying the procedure above yields $\lc^*(0)=\lc(0)\approx 1.65$, in
agreement with the accepted value $\lc^{\mathrm{cp}}\approx
1.6489$~\citep{Liggett1999}. This validates both the implementation and
the extrapolation method.
We note that $\lambda_c(\alpha)$ is estimated from two observables with
independent sources of randomness, the total range $N$ and the origin
occupation time $T$, whose polynomial extrapolations agree across the
range of $\alpha$ for which both are usable (Section~\ref{sec:crs-quadratic}). Since
$N$ and $T$ are accumulated by unrelated mechanisms, a counting or
bookkeeping artifact in one would not be reproduced in the other, so this
agreement provides a consistency check on the critical-value estimates
beyond the contact-process validation above.

\subsubsection{Diagnostics and reproducibility}

For each $(\lambda,\alpha)$ grid point we recorded the number of runs
that hit the event cap $M$ (``censored'' runs) and the mean number of
lattice expansions per run. In the \CRS{} regime
$\alpha\leq 5$ and in all \MT{} and \MTo{} scans the censored fraction
was effectively zero at subcritical $\lambda$. For the \CRS{} runs at
$\alpha\in\{10,20\}$ near criticality, a small fraction of runs
($\lesssim 5\%$) hit $M=10^{8}$; we flag these cases in the final
tables. Expansion counts ranged from 0 (deeply subcritical) to
$\sim 10$ per run near criticality, which sits well within the
size-safety margins of the dynamic lattice.

All simulations were performed in \textsf{R}, with the critical-regime
\MT{} scans accelerated using a compiled C\texttt{++} inner loop via
\textsf{Rcpp}. Random numbers were drawn from \textsf{R}'s Mersenne
Twister with explicit seeds recorded per scan. Typical scan parameters were $n = 50$ runs per $\lambda$-value
for the CRS and MT scans (raised to $n = 100$--$200$ for the extrapolation and
survival-probability figures), $n = 20$--$30$ near very slow parameter regions, and polynomial degrees 3--7 for the $1/\E[N]$ extrapolation. The simulation code will be made available on a public repository at the time of publication.

\subsection{Numerical results and conjectures}\label{sec:results}

We report results for the three models in turn.

\subsubsection{The \MTo{} model: linear mean range}

In the original Maki--Thompson model, stiflers are permanent. Every
site the rumor reaches becomes, after an exponential delay, a stifler
that never clears. On $\Z$, the front therefore advances into a purely
ignorant region on one side and into a permanent stifler wall on the
other (or on both sides, for a two-sided front started from a single
spreader at the origin).

In our simulations the rumor died in every realization we ran, for
$\lambda$ spanning more than two decades. Moreover, the sample mean
total range $\widehat{\E[N]}$ agrees with the exact value $2(\lambda+1)$ of
Theorem~\ref{prop:MToZ} to within two Monte Carlo standard errors at every
$\lambda$ simulated; see Table~\ref{tab:mto} and Figure~\ref{fig:mto}. The additive constant is resolved
at small $\lambda$, where $\widehat{\E[N]}/\lambda$ sits clearly above $2$
(for instance $2.20$ at $\lambda = 10$), and becomes negligible at large
$\lambda$, where the leading $2\lambda$ behavior dominates.

\begin{table}[h]
\centering
\caption{\MTo{} model on $\Z$: sample mean total range $\widehat{\E[N]}$,
estimated from $n$ independent runs at each $\lambda$, against the exact
prediction $2(\lambda+1)$ of Theorem~\ref{prop:MToZ}. The reported standard
errors are of the mean; $\widehat{\E[N]}$ agrees with $2(\lambda+1)$ to within
two standard errors at every $\lambda$.}
\label{tab:mto}
\begin{tabular}{rrrrrr}
\toprule
$\lambda$ & $n$ & $\widehat{\E[N]}$ & SE & $2(\lambda+1)$ & $\widehat{\E[N]}/\lambda$ \\
\midrule
$10$    & $20000$ & $22.0$    & $0.1$   & $22$     & $2.20$ \\
$50$    & $8000$  & $100.7$   & $0.8$   & $102$    & $2.01$ \\
$100$   & $4000$  & $204.6$   & $2.2$   & $202$    & $2.05$ \\
$500$   & $1000$  & $992.0$   & $21.6$  & $1002$   & $1.98$ \\
$1000$  & $300$   & $2067.1$  & $79.1$  & $2002$   & $2.07$ \\
$5000$  & $100$   & $9997.6$  & $770.4$ & $10002$  & $2.00$ \\
\bottomrule
\end{tabular}
\end{table}

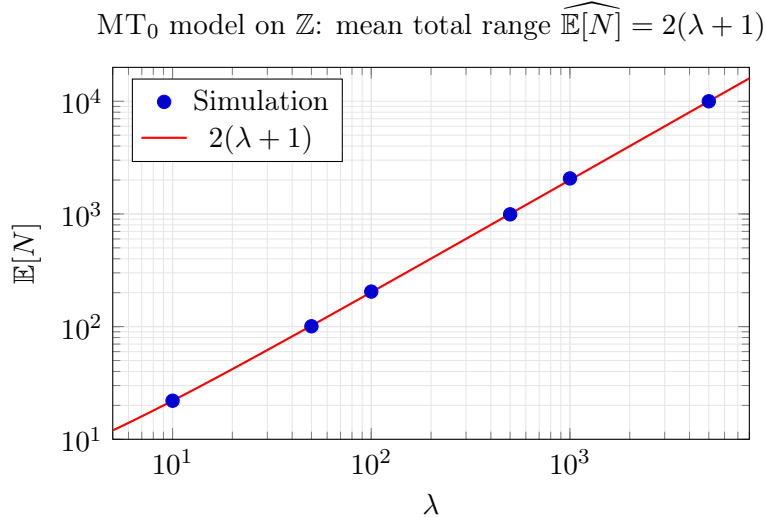
\begin{figure}[h]
\centering
\begin{tikzpicture}
\begin{axis}[
  width=10cm, height=6.5cm,
  xlabel={$\lambda$}, ylabel={$\E[N]$},
  xmode=log, ymode=log,
  xmin=5, xmax=8000, ymin=10, ymax=20000,
  grid=both, grid style={line width=.1pt, draw=gray!20},
  legend pos=north west,
  title={\MTo{} model on $\Z$: mean total range $\widehat{\E[N]} = 2(\lambda+1)$}
]
\addplot+[only marks, mark=*, mark size=2.5pt, color=blue]
  coordinates {
    (10,22.0) (50,100.7) (100,204.6)
    (500,992.0) (1000,2067.1) (5000,9997.6)
  };
\addlegendentry{Simulation}
\addplot[domain=5:8000, samples=100, red, thick]
  {2*(x+1)};
\addlegendentry{$2(\lambda+1)$}
\end{axis}
\end{tikzpicture}
\caption{Sample mean total range $\widehat{\E[N]}$ of the \MTo{} model on $\Z$
as a function of $\lambda$, averaged over $n$ independent runs per $\lambda$-value
(Table~\ref{tab:mto}). The data closely track the exact law $\E[N] = 2(\lambda+1)$
of Theorem~\ref{prop:MToZ} (red).}
\label{fig:mto}
\end{figure}

The factor $2$ admits a simple heuristic interpretation in terms of
independent left and right fronts. Consider the right-moving front.
When a spreader at site $x$ has just been created by its left neighbor
$x-1$, the new spreader has one spreader neighbor (the parent, at
$x-1$) and one ignorant neighbor (at $x+1$). Two competing transitions
affect its future: it can spread further at rate $\lambda$ (infecting
$x+1$), or it can be stifled by its parent at rate~$1$. The probability
that the front advances one more step is therefore approximately
$\lambda/(\lambda+1)$, and, once advanced, the same dynamics repeats at
the new front. The expected number of right-side sites reached is
thus approximately
\[
\sum_{k\geq 0}\Bigl(\tfrac{\lambda}{\lambda+1}\Bigr)^{k}
\;=\; \lambda+1,
\]
and by symmetry the left-side front contributes the same. Taken
together, the total expected range is $\approx 2(\lambda+1)\approx 2\lambda$
for large $\lambda$, which matches the observed slope.

The simulations thus confirm Theorem~\ref{prop:MToZ} numerically: the rumor
died in every run, and the measured mean range tracks the exact value
$2(\lambda+1)$, whose leading term is $2\lambda$.

\subsubsection{The \CRS{} model: a quadratic phase boundary} \label{sec:crs-quadratic}

Table~\ref{tab:crs} lists the estimated critical values $\lc^{*}(\alpha)$
obtained by polynomial extrapolation of $1/\E[N]$, for thirteen values of
$\alpha$ ranging from $0$ to $20$. The corresponding phase diagram is
shown in Figure~\ref{fig:crs_phase}.

\begin{table}[h]
\centering
\caption{Estimated critical values $\lc^*(\alpha)$ for the \CRS{} model on $\Z$, and the ratio $\lc^*(\alpha)/\alpha^2$ for $\alpha\geq 1$.}
\label{tab:crs}
\begin{tabular}{rrrr}
\toprule
$\alpha$ & $\lc^{*}$ & $\lc^*-\lc^*(0)$ & $\lc^*/\alpha^2$ \\
\midrule
$0$    & $1.65$ & $0.00$ & --     \\
$0.2$  & $2.5$  & $0.85$ & --     \\
$0.25$ & $2.66$ & $1.01$ & --     \\
$0.3$  & $2.9$  & $1.25$ & --     \\
$0.4$  & $3.6$  & $1.95$ & --     \\
$0.5$  & $4.06$ & $2.41$ & --     \\
$0.6$  & $4.7$  & $3.05$ & --     \\
$1$    & $8.15$ & $6.50$ & $8.15$ \\
$2$    & $25.5$ & $23.9$ & $6.38$ \\
$3$    & $62$   & $60.4$ & $6.89$ \\
$5$    & $191$  & $189$  & $7.64$ \\
$10$   & $713$  & $711$  & $7.13$ \\
$20$   & $2750$ & $2748$ & $6.87$ \\
\bottomrule
\end{tabular}
\end{table}

For $\alpha \geq 1$ the ratio $\lc^*(\alpha)/\alpha^2$ clusters in a narrow band around 7, and a least-squares fit of the form $\lc^*(\alpha) = \lc^*(0) + c\,\alpha^2$ to the $\alpha \geq 1$ points yields $c \approx 7$. All of the sub-unit $\alpha$ values lie above this law, and increasingly so as $\alpha$ decreases: for example $\lc^*(0.6) \approx 4.7$ versus the predicted $1.65 + 7(0.6)^2 \approx 4.17$, and $\lc^*(0.25) \approx 2.66$ versus $2.09$. The log-log plot in Figure~\ref{fig:crs_loglog} shows this directly: the slope-2 reference line tracks the data for $\alpha \geq 1$, while the smallest-$\alpha$ points lie above it. This is consistent with the quadratic scaling being an $\alpha \to \infty$ asymptotic rather than an all-$\alpha$ fit. We therefore propose:

\begin{conjecture}[Quadratic phase boundary]\label{conj:main}
For the \CRS{} model on $\Z$, the critical transmission rate satisfies
\[
\lc^*(\alpha) \;=\; \lc^*(0) + c\,\alpha^{2} + o(\alpha^{2})\qquad
(\alpha\to\infty),
\]
with $\lc^*(0)\approx 1.6489$ the critical value of the one-dimensional
contact process and $c\approx 7$.
\end{conjecture}

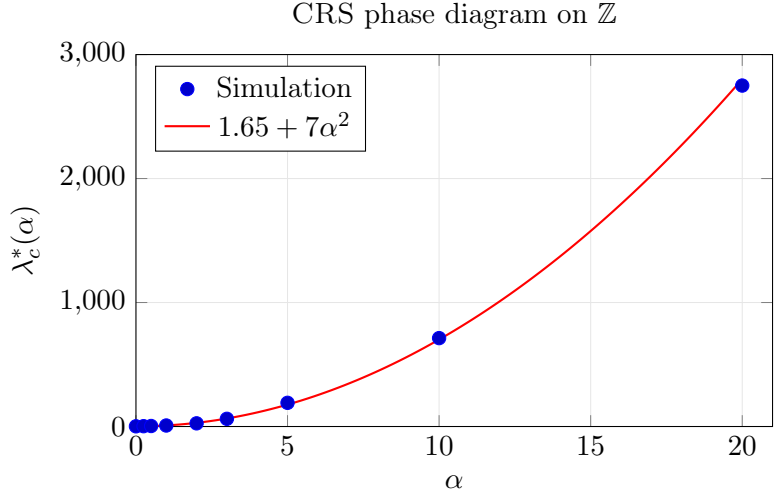
\begin{figure}[h]
\centering
\begin{tikzpicture}
\begin{axis}[
  width=10cm, height=6.5cm,
  xlabel={$\alpha$}, ylabel={$\lc^*(\alpha)$},
  xmin=0, xmax=21, ymin=0, ymax=3000,
  grid=both, grid style={line width=.1pt, draw=gray!20},
  legend pos=north west,
  title={\CRS{} phase diagram on $\Z$}
]
\addplot+[only marks, mark=*, mark size=2.5pt, color=blue]
  coordinates {
    (0,1.65) (0.25,2.66) (0.5,4.06) (1,8.15) (2,25.5)
    (3,62) (5,191) (10,713) (20,2750)
  };
\addlegendentry{Simulation}
\addplot[domain=0:20, samples=100, red, thick]
  {1.65 + 7*x^2};
\addlegendentry{$1.65 + 7\alpha^2$}
\end{axis}
\end{tikzpicture}
\caption{Phase diagram of the \CRS{} model on $\Z$. Blue markers are
Monte Carlo estimates of $\lc^*(\alpha)$; the red curve is the proposed
quadratic law of Conjecture~\ref{conj:main}.}
\label{fig:crs_phase}
\end{figure}

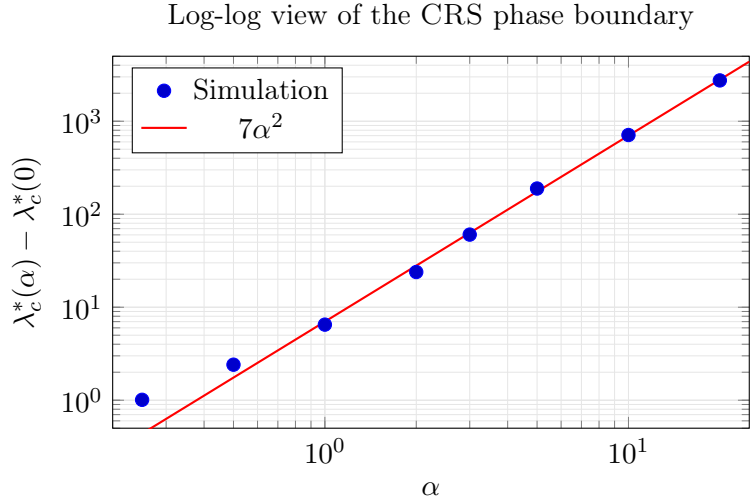
\begin{figure}[h]
\centering
\begin{tikzpicture}
\begin{loglogaxis}[
  width=10cm, height=6.5cm,
  xlabel={$\alpha$}, ylabel={$\lc^*(\alpha)-\lc^*(0)$},
  xmin=0.2, xmax=25, ymin=0.5, ymax=5000,
  grid=both, grid style={line width=.1pt, draw=gray!20},
  legend pos=north west,
  title={Log-log view of the \CRS{} phase boundary}
]
\addplot+[only marks, mark=*, mark size=2.5pt, color=blue]
  coordinates {
    (0.25,1.01) (0.5,2.41) (1,6.50) (2,23.9)
    (3,60.4) (5,189) (10,711) (20,2748)
  };
\addlegendentry{Simulation}
\addplot[domain=0.2:25, samples=50, red, thick] {7*x^2};
\addlegendentry{$7\alpha^2$}
\end{loglogaxis}
\end{tikzpicture}
\caption{Log-log plot of $\lc^*(\alpha)-\lc^*(0)$ versus $\alpha$ for the
\CRS{} model. The slope-$2$ reference line $7\alpha^{2}$ tracks the data for $\alpha \geq 1$; the smallest $\alpha$ points lie above it, consistent with the law being asymptotic.}
\label{fig:crs_loglog}
\end{figure}

\paragraph{Cross-validation via $\E[T]$.}
Applying the same polynomial-extrapolation procedure to $1/\E[T]$
yields critical-value estimates consistent with those from
$1/\E[N]$ at all $\alpha$ for which both observables are usable.
The $\E[T]$ estimator is noisier near criticality, and for very large
$\alpha$ (where the origin spends only brief intervals in state~1)
it becomes unreliable; we therefore use $\E[N]$ as our primary estimator and report $\E[T]$-based values as a consistency check.

\paragraph{A heuristic for the quadratic law.}
The mean-field theory fails to capture the $\alpha$-dependence of $\lc^*$, so the mechanism driving the quadratic law must be spatial. We offer the following heuristic at the level of a single spreading event.

Consider a spreader at site $x$ that has just infected its previously
ignorant neighbor $y=x\pm 1$. Immediately after this event both $x$ and
$y$ are spreaders, each with the other as a spreader neighbor. At the
resulting spreader--spreader edge, three competing transitions may fire
next:

\begin{itemize}
\item $x$ (or $y$) advances the front by infecting a new ignorant neighbor,
      at rate $\lambda$;
\item $x$ is stifled by $y$ (or vice versa), at rate $\alpha$;
\item $x$ or $y$ forgets the rumor, at rate~$1$.
\end{itemize}

For the front to continue advancing, the spread event must occur before
either of the two sites is removed from state~1. For $\alpha\gg 1$ the
probability of stifling dominates over forgetting, and the probability
of advancing the front before stifling at a single edge is approximately
$\lambda/(2\alpha+\lambda)\sim \lambda/(2\alpha)$. Crucially, both sites
$x$ and $y$ are simultaneously at risk: stifling is a \emph{two-body}
event in which the newly formed parent--child pair can extinguish each
other. The damage done per unit of ``activity'' therefore scales as
$\alpha\cdot\alpha=\alpha^{2}$, and for the rumor to survive the
spreading rate must compensate, giving $\lc^*\sim c\,\alpha^{2}$.

A rigorous implementation of this intuition would presumably proceed
by refining the oriented-percolation comparison in
\cite[Theorem~2.2(a)]{CRS2012} to track the quadratic cost. We leave this
as an open problem.

\subsubsection{The \MT{} model: a rapidly growing critical line}

For the \MT{} variant, stiflers contribute to stifling on the same
footing as spreaders. Even though stiflers still forget at rate~$1$, the
\emph{stifler wake} left behind the front creates back-pressure on the
leading spreader. Intuitively, as $\alpha$ grows the wake clears more
slowly relative to its rate of creation, and the critical transmission
rate $\lc(\alpha)$ must grow rapidly in order for the front to advance
through the wake.

Numerically, the 14 measured data points are consistent with a phase
transition at every value of $\alpha$ we could resolve, ranging from
$\alpha=0.10$ up to $\alpha=0.85$, with estimated $\lambda_c$ rising from $\approx 2.1$ to $\approx 430$ over that interval (Table~\ref{tab:mt} and Figure~\ref{fig:mt-phase}); Figure~\ref{fig:mt_invEN} illustrates the underlying $1/\E[N]$ extrapolation for four representative values of $\alpha$.
A least-squares fit of these points to the two-parameter family
\begin{equation} \label{eq:mt-fit}
\lambda_c(\alpha) = \lambda_c(0) + c \cdot \dfrac{\alpha}{(1-\alpha)^{q}}
\end{equation}
yields $c \approx 2.8$ and $q \approx 2.65$. The fit reproduces the steep,
faster-than-quadratic growth of $\lambda_c(\alpha)$ across the measured range,
with relative deviations of up to about $20\%$. Because the extrapolated
critical values carry substantial, only partially quantified uncertainty,
particularly at the largest $\alpha$, we regard~\eqref{eq:mt-fit} as an
empirical interpolation within our measurement window rather than a tested
scaling law. At $\alpha=0.9$ we were able to push to
$\lambda=850$ with $\E[N]/\lambda$ climbing from $\approx 16$ at
$\lambda=500$ to $\approx 26$ at $\lambda=850$, consistent with an
approach to a critical value near the formula prediction
$\lc(0.9)\approx 1126$ but not conclusive on its own. This persistence of the transition at $\alpha=0.9$, at a large but finite rate, is consistent with Theorem~\ref{thm:phase}.

The steep rise of $\lc(\alpha)$ as $\alpha$ approaches $1$ admits a simple
rate-balance reading: a spreader is converted into a (still-stifling) stifler at
rate $\alpha$ per informed neighbor, while informed sites are cleared by forgetting
at rate $1$, so near $\alpha=1$ stifler creation begins to outpace forgetting, the
wake clears ever more slowly, and $\lc(\alpha)$ rises sharply. We
stress that the mean-field system~\eqref{eq:MT-ODE} does not itself single out
$\alpha=1$: its rumor-free equilibrium is unstable iff $\lambda>1$ for every
$\alpha$, so the threshold here is a spatial effect rather than a mean-field
prediction.

\begin{table}[h]
\centering
\caption{Estimated critical values $\lc(\alpha)$ for the \MT{} model on $\Z$.
The value at $\alpha=0.9$ marked $^{\dagger}$ is an extrapolation from the
fit (\ref{eq:mt-fit}); the measurement at $\lambda=850$ shows
$\E[N]/\lambda\approx 26$ and is consistent with a nearby critical value.}
\label{tab:mt}
\small
\begin{tabular}{l*{14}{c}}
\toprule
$\alpha$    & $0.10$ & $0.15$ & $0.20$ & $0.25$ & $0.30$ & $0.35$ & $0.40$
            & $0.45$ & $0.50$ & $0.60$ & $0.70$ & $0.75$ & $0.80$ & $0.85$ \\
\midrule
$\lc^{\MT}$ & $2.1$ & $2.4$ & $2.8$ & $3.2$ & $3.7$ & $4.5$ & $5.5$
            & $7$ & $9$ & $17$ & $45$ & $80$ & $180$ & $430$ \\
\bottomrule
\end{tabular}
\vspace{0.5em}\\
Formula extrapolation: $\lc(0.9)\approx 1126^{\dagger}$,
$\lc(0.95)\approx 7450^{\dagger}$.
\end{table}

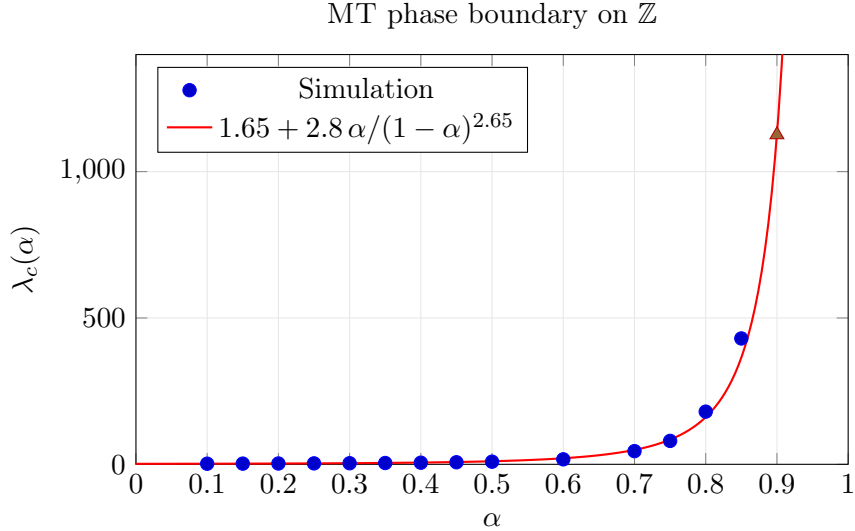
\begin{figure}[h]
\centering
\begin{tikzpicture}
\begin{axis}[
  width=11cm, height=7cm,
  xlabel={$\alpha$}, ylabel={$\lc(\alpha)$},
  xmin=0, xmax=1, ymin=0, ymax=1400,
  grid=both, grid style={line width=.1pt, draw=gray!20},
  legend pos=north west,
  title={\MT{} phase boundary on $\Z$}
]
\addplot+[only marks, mark=*, mark size=2.5pt, color=blue]
  coordinates {
    (0.10,2.1) (0.15,2.4) (0.20,2.8) (0.25,3.2) (0.30,3.7)
    (0.35,4.5) (0.40,5.5) (0.45,7) (0.50,9) (0.60,17)
    (0.70,45) (0.75,80) (0.80,180) (0.85,430)
  };
\addlegendentry{Simulation}
\addplot[domain=0:0.93, samples=200, red, thick]
  {1.65 + 2.8*x/(1-x)^2.65};
\addlegendentry{$1.65 + 2.8\,\alpha/(1-\alpha)^{2.65}$}
\addplot+[only marks, mark=triangle*, mark size=3pt, color=red!70!black,
          forget plot]
  coordinates { (0.90, 1126) };
\end{axis}
\end{tikzpicture}
\caption{Phase diagram of the \MT{} model on $\Z$. Blue circles are
the 14 measured values of $\lc(\alpha)$ from Table~\ref{tab:mt}; the red
curve is the least-squares fit of (\ref{eq:mt-fit}) with $c\approx 2.8$
and $q\approx 2.65$, growing steeply as $\alpha\uparrow 1$. The red triangle
at $\alpha=0.9$ marks the fit extrapolation $\lc(0.9)\approx 1126$. By
Theorem~\ref{thm:phase} the critical line stays finite for all $\alpha>0$, so the
fit is a local interpolation and the steepening near $\alpha=1$ is not a genuine divergence.}
\label{fig:mt-phase}
\end{figure}

\begin{figure}[t]
\centering
\includegraphics[width=\textwidth]{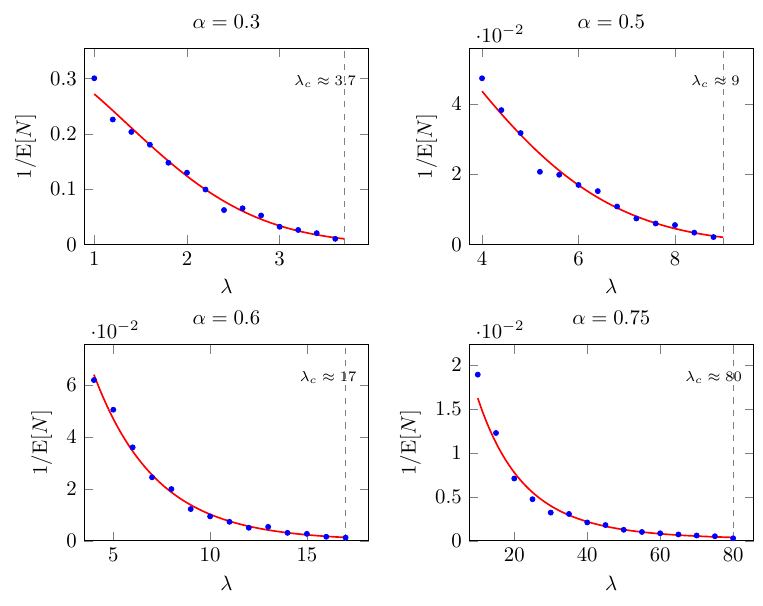}
\caption{Estimation of the \MT{} critical transmission rate by polynomial extrapolation of
$1/\E[N]$ (Section~\ref{sec:method}). Each panel shows the measured $1/\E[N]$ (blue markers)
against $\lambda$ at a fixed stifling rate $\alpha$ ($n=100$ runs per $\lambda$); the solid curve
is a smooth guide to the eye and the dashed line marks the estimated critical value $\lc(\alpha)$
at which $1/\E[N]$ extrapolates to zero. The estimates span more than an order of magnitude, from
$\lc\approx 3.7$ at $\alpha=0.3$ to $\lc\approx 80$ at $\alpha=0.75$, in agreement with Table~\ref{tab:mt}.}
\label{fig:mt_invEN}
\end{figure}

\begin{conjecture}[\MT{} model]\label{conj:mt}
For the \MT{} model on $\Z$, the critical transmission rate $\lc(\alpha)$, finite
for every $\alpha>0$ by Theorem~\ref{thm:phase}, is increasing in $\alpha$ and
satisfies $\lc(\alpha)\to\infty$ as $\alpha\to\infty$, growing far faster than the
quadratic \CRS{} law. Over the measured range $0.1\le\alpha\le 0.85$ it is well
described by the empirical interpolation
\[
\lc(\alpha)\;\approx\;\lc(0)+c\,\frac{\alpha}{(1-\alpha)^{q}},\qquad
c\approx 2.8,\ q\approx 2.65 .
\]
The sharp steepening near $\alpha=1$ reflects the rate-balance crossover described
above, not a singularity: by Theorem~\ref{thm:phase} the critical line stays finite
at and beyond $\alpha=1$, so the apparent pole of the fit at $\alpha=1$ is a feature
of the interpolation within our measurement window rather than a genuine divergence.
\end{conjecture}

Having estimated both critical lines, we can compare them directly. Figure~\ref{fig:crs_vs_mt}
overlays $\lambda^*_c(\alpha)$ for the CRS model and $\lambda_c(\alpha)$ for the MT model over
the common range $\alpha \in [0, 0.6]$. Both start from the contact-process value
$\lambda^{cp}_c \approx 1.6489$ at $\alpha = 0$, and the CRS line lies below the MT line at every $\alpha$ we measured, with the gap widening as $\alpha$ increases. This is the numerical evidence for the ordering $\lambda^*_c(\alpha) \le \lambda_c(\alpha)$ conjectured in Section~\ref{sec:theorems}.
\begin{figure}[t]
\centering
\includegraphics[width=0.75\textwidth]{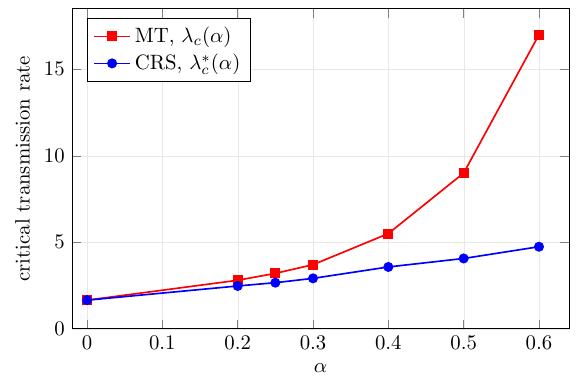}
\caption{Estimated critical transmission rates of the \CRS{} model ($\lambda_c^{*}(\alpha)$, blue) and the \MT{} model ($\lambda_c(\alpha)$, red) on $\mathbb{Z}$ over $\alpha\in[0,0.6]$. Both lines start
from the common contact-process value $\lambda_c^{cp}\approx 1.6489$ at $\alpha=0$; the CRS line lies
below the MT line at every $\alpha$, and the gap widens as $\alpha$ increases. This is the numerical
evidence for the ordering $\lambda_c^{*}(\alpha)\le\lambda_c(\alpha)$ conjectured in Section~\ref{sec:theorems}}
\label{fig:crs_vs_mt}
\end{figure}

\subsection{Survival probability as a function of $\lambda$}\label{sec:survival}

The definition of $\lc(\alpha)$ as the value of $\lambda$ at which $\E[N]$ diverges implicitly assumes that the survival probability $\theta(\lambda,\alpha)$ is non-decreasing in $\lambda$ for
each fixed $\alpha$. For the basic contact process this is a classical consequence of Harris's
coupling \cite{Liggett1999}; for the \CRS{} and \MT{} models the same property is expected but,
to our knowledge, has not been rigorously established.

As an empirical estimate of $\theta(\lambda,\alpha)$ we use the fraction of independent runs that
reach $N \geq N_{\mathrm{thr}} = 500$ distinct visited sites before extinction or the event cap.
We performed $n = 200$ runs at each of twelve $\lambda$-values, for both the \CRS{} and \MT{}
models at $\alpha = 0.5$ (Figure~\ref{fig:survival}). The estimate is non-decreasing in $\lambda$
at every consecutive pair of grid points for both models, providing numerical support for the
expected monotonicity, and the smooth sigmoidal shape is consistent with a phase transition at a
finite $\lc(0.5)$ for each model.

\begin{figure}[t]
\centering
\includegraphics[width=0.8\textwidth]{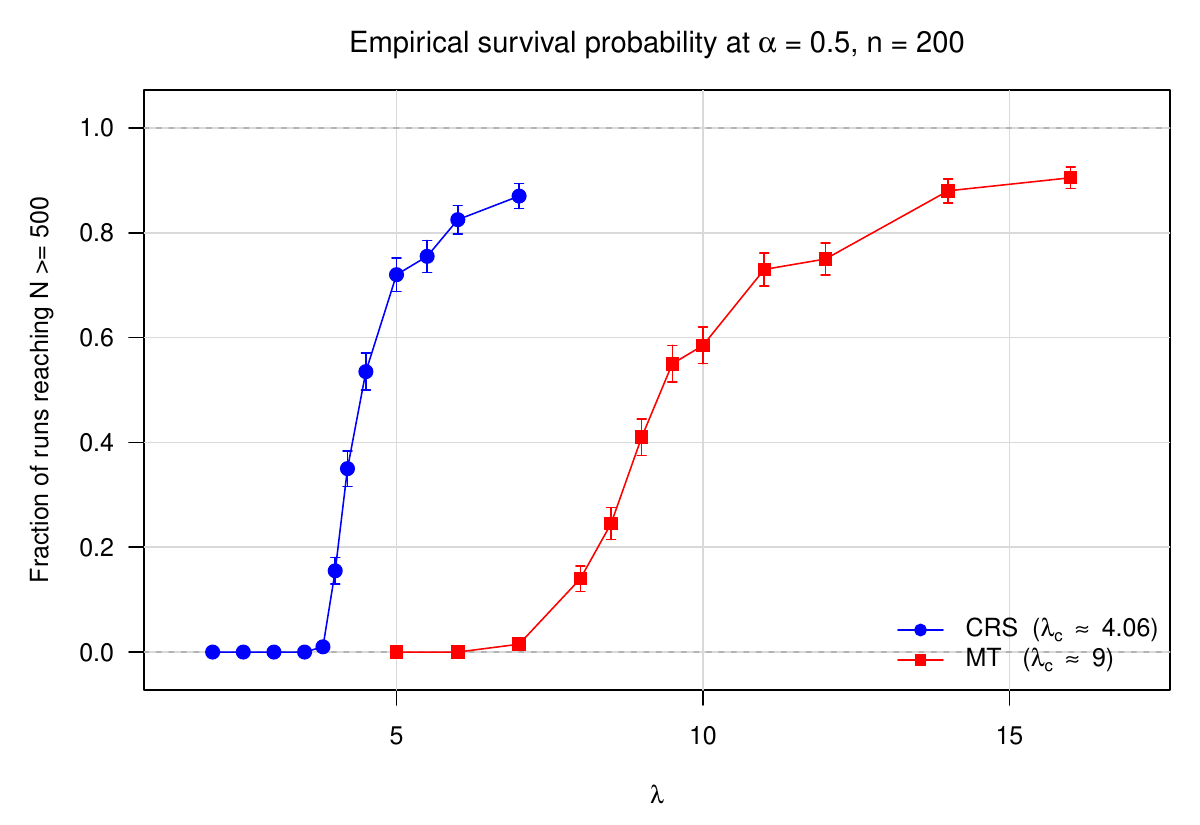}
\caption{Empirical survival probability $\widehat\theta(\lambda,\alpha)$ for the \CRS{} (blue)
and \MT{} (red) models at $\alpha = 0.5$, defined as the fraction of $n = 200$ runs that reach
$N_{\mathrm{thr}} = 500$ distinct visited sites before extinction or the event cap. Error bars are
one binomial standard error. Both curves are non-decreasing in $\lambda$ at every consecutive pair
of grid points, consistent with the expected (but not yet rigorously established) monotonicity of
$\theta(\lambda,\alpha)$ in $\lambda$, and with a phase transition at a finite $\lc$ for each model.}
\label{fig:survival}
\end{figure}

\section{Discussion}\label{sec:discussion}
The results obtained in this work highlight the fundamental role played by the interaction between spatial structure and local rumor dynamics in determining the long-term behavior of stochastic spreading processes. Although the three models considered here share the same underlying transmission mechanism, they exhibit remarkably different qualitative behaviors. In particular, the presence or absence of mechanisms that remove active spreaders from the rumor front leads to substantially different phase diagrams. The classical Maki--Thompson model becomes extinct for every transmission rate, whereas the introduction of forgetting mechanisms gives rise to genuine phase transitions between extinction and survival. This contrast demonstrates how subtle modifications of local interaction rules can produce significant changes in the global behavior of spatial stochastic systems.

Our analysis shows that, in one dimension, the linear geometry of the lattice imposes strong constraints on the propagation of the rumor, and the interactions occurring near the expanding boundary play a crucial role in determining whether the process can sustain long-range growth. For the \CRS{} model, the quadratic growth of the critical transmission rate for large stifling rates suggests that the main obstruction to survival arises from repeated stifling events occurring at the rumor frontier. The heuristic argument based on parent--child interactions indicates that this suppression mechanism is not merely a local phenomenon but is amplified by the spatial organization of the process, leading to the observed scaling behavior.

The comparison between the \CRS{} and \MT{} models further illustrates that the precise formulation of the stifling mechanism has a decisive impact on the critical behavior. Although both models incorporate forgetting and admit survival for sufficiently large transmission rates, their critical curves exhibit distinct asymptotic regimes. In the \MT{} model, the steep growth of the critical transmission rate as $\alpha$ approaches one reflects the increasing difficulty of maintaining a population of active spreaders when stifling dominates the dynamics. These observations emphasize that the microscopic details of the interaction rules, rather than only macroscopic parameters such as transmission and forgetting rates, are essential in determining the emergence of phase transitions.

Several interesting questions remain open even in the one-dimensional setting considered here. A natural direction is the rigorous investigation of the monotonicity properties of survival probabilities with respect to the model parameters, as well as a more detailed comparison between the different rumor mechanisms. Establishing ordering relations between the models and determining sharper asymptotic properties of the critical curves represent challenging mathematical problems that may provide further insight into the role of local interactions in spatial spreading processes.

Beyond one dimension, the present results suggest several avenues for future research. Extending these models to higher-dimensional lattices would allow the investigation of how the geometry of the rumor frontier influences critical behavior and survival mechanisms. Similarly, the study of tree-like structures and heterogeneous networks may reveal how branching effects and complex connectivity patterns compete with local extinction mechanisms. Such extensions could provide a bridge between rigorous results for spatial stochastic systems and the analysis of information spreading phenomena in realistic networks.

Finally, the numerical results obtained in this work illustrate the value of combining rigorous probabilistic methods with computational approaches in the study of critical phenomena. Exact continuous-time simulations provide access to parameter regimes where analytical techniques become challenging, while theoretical arguments help identify the mechanisms underlying the observed scaling laws. This interplay between analysis and computation offers a promising framework for further investigations of phase transitions in stochastic models of information spreading.


\section*{Acknowledgments}

This work has been developed with the support of Conselho Nacional de Desenvolvimento Científico e Tecnológico - CNPq (Grant 316121/2023-1), and Fundação de Amparo à Pesquisa do Estado de São Paulo - FAPESP (Grant 23/13453-5). This work was also supported by the National Institute of Science and Technology in Stochastic Modelling and Complexity (INCT-NUMEC), with funding from CNPq (grant no. 408590/2024-6).


\end{document}